\documentclass{article}
\usepackage[utf8]{inputenc}
\usepackage{amsmath,amsthm,amssymb,verbatim,xcolor}
\newtheorem{theorem}{Theorem}[section]
\newtheorem{corollary}{Corollary}[section]
\title{Bahadur efficiency of EDF based normality tests when parameters are estimated}
\author{Bojana Milo\v sevi\'c\thanks{University of Belgrade -- Faculty of Mathematics, \texttt{bojana@matf.bg.ac.rs}}, Ya.Yu. Nikitin\thanks{Department of Mathematics and Mechanics, Saint-Petersburg State University,
National Research University - Higher School of Economics, Russia}, Marko Obradovi\'c\thanks{University of Belgrade -- Faculty of Mathematics, \texttt{marcone@matf.bg.ac.rs}}}

\date{}

\begin{document}

\maketitle
\begin{abstract}
    In this paper some well-known tests based on empirical distribution functions (EDF) with estimated parameters for testing composite normality hypothesis are revisited,  and some new results on asymptotic properties are provided.  In particular, the approximate Bahadur slopes are obtained --- in the case of close alternatives --- for the EDF-based tests as well as the likelihood ratio test. The local approximate efficiencies are calculated for several close alternatives. The obtained results could serve as a  benchmark for evaluation of the quality of recent and future normality tests.
\end{abstract}
{keywords: asymptotic efficiency, goodness-of-fit, composite hypothesis}
\section{Introduction}

For testing the goodness-of-fit (GOF) null hypothesis that the sample is taken from a fully specified continuous distribution $F_0$, the predominantly used tests in practice are those based on some distance between the empirical distribution function (EDF) $F_n$ and $F_0$.

The most widely used EDF-based tests is the well-known Kolmogorov-Smirnov test \cite{kolmogorov1933} with statistic $D_n=\sup_{x}|F_n(x)-F_0(x)|$  based on the $L^\infty$ distance. Other popular tests include the Cramer--von Mises \cite{darling1957} and Anderson--Darling \cite{anderson1952} test based on the weighted $L^2$ distance between $F_n$ and $F_0$. Different variations of these test statistics exist. Watson proposed the centered versions of the Kolmogorov-Smirnov \cite{watson1976} (see also \cite{darling1983a,darling1983b}) and the Cramer--von Mises \cite{watson1961} tests. Other variants were proposed by Kuiper \cite{kuiper1960} and Khmaladze \cite{khmaladze1982} among others.

The properties of EDF-based tests are well-known. All these tests are distribution-free under the null hypothesis makes them omnibus GOF tests applicable regardless of $F_0$. Their asymptotic distributions follow from the limiting process of $F_n(t)-F_0(t)$ when $n\to\infty$, which is the Brownian bridge. Large deviations of these statistics are available in \cite{nikitinKnjiga}.

However, more often than not, we would like to test a composite GOF null hypothesis that the sample comes from a family of distributions $F_0(x;\theta)$ indexed by a finite-dimensional parameter $\theta$. In this scenario, we need to estimate $\theta$ in order to apply the EDF-based tests.  The problem is that the tests are no longer distribution-free, and their distribution depends on $F_0$ and $\theta$. 

In case of location-scale families, it can be easily shown that the distribution does not depend on the location and scale parameters, but only on $F_0$. Therefore in this case we can consider GOF tests for particular null location-scale families of distributions such as normal, exponential, logistic, Cauchy, etc.

 The modified EDF-based tests  have been proposed and/or their properties investigated by Durbin \cite{darling1955cramer}, Kac, Kiefer and Wolfowitz \cite{kac1955tests}, Lilliefors \cite{lilliefors1967kolmogorov,lilliefors1969kolmogorov},  Sukhatme \cite{sukhatme1972},  and the asymptotic theory have been studied by Durbin \cite{durbin1973weak} and Stephens \cite{stephens1976asymptotic}, among others.
 
 A popular tool for asymptotic comparison of tests is the Bahadur asymptotic efficiency. One of the advantages over other types of efficiencies is that is more convenient when the asymptotic distributions are not normal. A comprehensive review of the Bahadur efficiencies of EDF-tests for the simple null hypothesis is available in \cite{nikitinKnjiga}.
 
The calculation of Bahadur efficiency is heavily dependent on the large deviation function of the test statistic, which is not available for the statistics with estimated parameters. An approach in this direction was done by  Arcones \cite{arcones2006bahadur}  for the case of Kolmogorov-Smirnov  normality test (also known as Lilliefors normality test), however, only upper
and lower estimates for large deviations were obtained in a very complicated form. The only test for which the Bahadur efficiencies were calculated  is the Kolmogorov-Smirnov exponentiality test \cite{nikitin2007lilliefors}. There the 
 corresponding large deviations were obtained using particular convenient properties of the exponential distribution.
 
 When large deviations are unavailable, a common way out is to use the so-called approximate Bahadur efficiency. Instead of the large deviations, its calculation requires only the tail behaviour of the asymptotic distribution.
 The quality of approximation has been shown to be good locally and for some statistics (e.g. U-statistics \cite{nikitinMetron,nikitin1999large} and their supremum \cite{nikitin2010large,milovsevic2016two}), exact and approximate Bahadur efficiency locally coincide.
 
 In this paper we compare EDF-based tests in terms of approximate Bahadur efficiency when testing the null normality hypothesis with both parameters unknown. In Section 2 we present the test statistics and their asymptotic behaviour and in Section 3 we calculate the efficiencies.
 
 \section{Test statistics}

Consider now the case of testing normality, i.e. the null hypothesis is $H_0:F(x)=\Phi(\frac{x-\mu}{\sigma})$, where $\Phi$ is the standard normal distribution function, and  unknown parameters $\mu$ and $\sigma$ are the mean and standard deviation.

The tests we consider are all based on difference $$\Delta_n(t;\hat{\mu},\hat{\sigma})=F_n(t)-\Phi\Big(\frac{t-\hat{\mu}}{\hat{\sigma}}\Big):$$

\begin{itemize}
    \item the Kolmogorov--Smirnov normality test with statistic
    \begin{align}
        D_n&=\sup_{t\in \mathbb{R}}\bigg|\Delta_n(t;\hat{\mu},\hat{\sigma})\bigg|;
    \end{align}
    \item the Cramer--von Mises normality test
    \begin{align}
        \omega^2_n&=\int_{-\infty}^{\infty}\Delta_n^2(t;\hat{\mu},\hat{\sigma})d\Phi\Big(\frac{t-\hat{\mu}}{\hat{\sigma}}\Big);
    \end{align}
    \item the Anderson--Darling normality test
    \begin{align}
       A^2_n&=\int_{-\infty}^{\infty} \frac{\Delta_n^2(t;\hat{\mu},\hat{\sigma})}{\Phi\Big(\frac{t-\hat{\mu}}{\hat{\sigma}}\Big)\bigg(1-\Phi\Big(\frac{t-\hat{\mu}}{\hat{\sigma}}\Big)\bigg)}d\Phi\Big(\frac{t-\hat{\mu}}{\hat{\sigma}}\Big);
    \end{align}
    \item the Watson--Darling variation of the Kolmogorov--Smirnov normality test
    \begin{align}
       G_n&=\sup_{t\in \mathbb{R}}\bigg|\Delta_n(t;\hat{\mu},\hat{\sigma})-\int_{-\infty}^{\infty}\Delta_n(z;\hat{\mu},\hat{\sigma})d\Phi\Big(\frac{z-\hat{\mu}}{\hat{\sigma}}\Big)dz\bigg|;
    \end{align}
    \item the Watson variation of the Cramer--von Mises normality test
  
   \begin{equation} U^2_n=\int_{-\infty}^{\infty}\bigg(\Delta_n(t;\hat{\mu},\hat{\sigma})-\int_{-\infty}^{\infty}\Delta_n(z;\hat{\mu},\hat{\sigma})d\Phi\Big(\frac{z-\hat{\mu}}{\hat{\sigma}}\Big)dz\bigg)^2 d\Phi\Big(\frac{t-\hat{\mu}}{\hat{\sigma}}\Big)dt,
   \end{equation}
\end{itemize}
where $\hat{\mu}=\bar{X}_n$ and $\hat{\sigma}^2=S^2$ are the maximum likelihood estimators of $\mu$ and $\sigma^2$.
To describe the asymptotic distribution of the test statistics, we define the following empirical processes:

\begin{align*}
    \eta_n(x;\mu,\sigma^2)&=F_n(\mu+\sigma x)-\Phi(x);\\
    \xi_n(x;\mu,\sigma^2)&=F_n(\mu+\sigma x)-\Phi(x)-\int_{-\infty}^{\infty}(F_n(\mu+\sigma x)-\Phi(z))\varphi(z)dz.
\end{align*}

Then, our statistics can be represented as 
\begin{align*}
    D_n&=\sup_{x\in  \mathbb{R}}|\eta_n(x;\hat{\mu},\hat{\sigma}^2)|;\\
    \omega^2_n&=\int_{-\infty}^{\infty} \eta_n^2(x;\hat{\mu},\hat{\sigma}^2)\varphi(x)dx;\\
    A^2_n&=\int_{-\infty}^{\infty} \frac{\eta_n^2(x;\hat{\mu},\hat{\sigma}^2)}{\Phi(x)(1-\Phi(x))}\varphi(x)dx;\\
    G_n&=\sup_{x\in  \mathbb{R}}|\xi_n(x;\hat{\mu},\hat{\sigma}^2)|;\\
U^2_n&=\int_{-\infty}^{\infty} \xi_n^2(x;\hat{\mu},\hat{\sigma}^2)\varphi(x)dx.
\end{align*}
It can be easily shown that all statistics are location and scale free under the null hypothesis of normality. Therefore, in what follows we assume that  true parameters are $\mu_0=0$ and $\sigma_0=1$.

\subsection{Asymptotic behaviour}

\begin{theorem} Let $X_1,X_2,...,X_n$ be an i.i.d. sample from normal $\mathcal{N}(0,1)$. Then the empirical processes $\sqrt{n}\eta_n(x;\hat{\mu},\hat{\sigma}^2)$ and $\sqrt{n}\xi_n(x;\hat{\mu},\hat{\sigma}^2)$ converge weakly in $D(\mathbb{R})$ to  centered Gaussian processes $\eta(x)$  and $\xi(x)$ whose covariance functions are respectively equal to
\begin{align*}
   K_{\eta}(x,y)&= \Phi(\min(x,y))-\Phi(x)\Phi(y)-\varphi(x)\varphi(y)-\frac{1}{2}xy\varphi(x)\varphi(y),\\
    K_{\xi}(x,y)&= \Phi(\min(x,y))-\Phi(x)\Phi(y)+\frac{1}{2}\Phi(x)(1-\Phi(x))+\frac{1}{2}\Phi(y)(1-\Phi(y))\\&+\frac{1}{2\sqrt{\pi}}(\varphi(x)+\varphi(y))-\varphi(x)\varphi(y)-\frac12xy\varphi(x)\varphi(y)+\frac{1}{12}-\frac{1}{4\pi}.
\end{align*}
\end{theorem}
\textbf{Proof}. For a fixed $x$, from \cite{randles1982asymptotic} we have the following representation:
\begin{align*}
    \sqrt{n}\eta_n(x;\hat{\mu},\hat{\sigma}^2)&=\sqrt{n}\eta_n(x;0,1)+\sqrt{n}\hat{\mu}\cdot\frac{\partial}{\partial \mu}{\mathbf{E}}\Big[{\rm I}\{X_1<\mu+\sigma x\}-\Phi(x)\Big]\Big|_{\mu=0,\sigma^2=1}\\&+\sqrt{n}(\hat{\sigma}^2-1)\cdot\frac{\partial}{\partial \sigma^2}{\mathbf{E}}\Big[{\rm I}\{X_1<\mu+\sigma x\}-\Phi(x)\Big]\Big|_{\mu=0,\sigma^2=1}\!+ o_P(1)\\&=\sqrt{n}\eta_n(x;0,1)+\varphi(x)\cdot\sqrt{n}\hat{\mu}+\frac x{2}\varphi(x)\cdot\sqrt{n}(\hat{\sigma}^2-1)+ o_P(1).
\end{align*}

From the multivariate central limit theorem it is straightforward to show  that the finite dimensional distributions are asymptotically normal.

The tightness of this process follows from the tightness property of the first summand (see \cite[Chapter 3]{billingsley}). The remaining components are just deterministic continuous functions of $x$ multiplied by a random variable, and, as such, tight in $C(\mathbb{R})$.

Taking into account the Bahadur represention of the estimator for $\sigma^2$, $$\hat{\sigma}^2-1=\frac{1}{n^2}\sum_{i,j}\frac{(X_i-X_j)^2}{2}-1=\frac{2}{n}\sum_{i}\frac{X_i^2-1}{2}+o_p(1),$$ we obtain that
the covariance function is

\begin{align*}
    K_\eta(x,y)&=K_0(x,y)+\varphi(y){\mathbf{E}}\Big[{\rm I}\{X< x\}X\Big] +\frac{y\varphi(y)}{2}{ \mathbf{E}}\Big[{\rm I}\{X< x\}(X^2-1)\Big]\\&+ \varphi(x){ \mathbf{E}}[{\rm I}\{X< y\}X]+ \frac{x\varphi(x)}{2}{\mathbf{E}}\Big[{\rm I}\{X< y\}(X^2-1)\Big] + \varphi(x)\varphi(y) \\&+ \frac{xy\varphi(x)\varphi(y)}{4}{ \mathbf{E}}\Big((X^2-1)^2\Big)
    \\&=K_0(x,y)-\varphi(y)\cdot\varphi(x) -\frac{y\varphi(y)}{2}\cdot x\varphi(x) -\varphi(x)\cdot\varphi(y) \\&-\frac{x\varphi(x)}{2}\cdot y\varphi(y) + \varphi(x)\varphi(y) + \frac{xy\varphi(x)\varphi(y)}{4}\cdot 2\\&
    =K_0(x,y)-\varphi(x)\varphi(y)-\frac{1}{2}xy\varphi(x)\varphi(y),
\end{align*}
where
\begin{align*}
   K_0(x,y)=\Phi(\min(x,y))-\Phi(x)\Phi(y)
\end{align*}
is the covariance function of the limiting process $\{\eta(x;0,1)\}$.

The same arguments for convergence of $\eta_n$ hold for the empirical process $\xi_n$, too, following its representation as
\begin{align*}
    \sqrt{n}\xi_n&(x;\hat{\mu},\hat{\sigma}^2)=\sqrt{n}\eta_n(x;0,1)+\frac12-\frac1n\sum_{i=1}^n\Phi(-X_i)+\sqrt{n}\hat{\mu}\\&\cdot\frac{\partial}{\partial \mu}{ \mathbf{E}}\Big[{\rm I}\{X_1<\mu+\sigma x\}-\Phi(x)-\Phi\Big(\frac{\mu-X_i}{\sigma}\Big)\Big]\Big|_{\mu=0,\sigma^2=1}+\sqrt{n}(\hat{\sigma}^2-1)\\&\cdot\frac{\partial}{\partial \sigma^2}{ \mathbf{E}}\Big[{\rm I}\{X_1<\mu+\sigma x\}-\Phi(x)-\Phi\Big(\frac{\mu-X_i}{\sigma}\Big)\Big]\Big|_{\mu=0,\sigma^2=1}+ o_P(1)\\&=\sqrt{n}\eta_n(x;0,1)+\frac1n\sum_{i=1}^n\Phi(X_i)-\frac12+(\varphi(x)-\frac{1}{2\sqrt{\pi}})\cdot\sqrt{n}\hat{\mu}\\&+\frac x{2}\varphi(x)\cdot\sqrt{n}(\hat{\sigma}^2-1)\!+\! o_P(1),
\end{align*}
while its covariance function is

\begin{align*}
K_{\xi}(x,y)&=K_0(x,y)+\frac{1}{2\sqrt{\pi}} (\phi(x)+\phi(y)-\frac{1}{\sqrt{\pi}})-\phi(x) (\phi(y)-\frac{1}{2\sqrt{\pi}})\\&-\phi(y) (\phi(x)-\frac{1}{2\sqrt{\pi}})+(\phi(x)-\frac{1}{2\sqrt{\pi}}) (\phi(y)-\frac{1}{2\sqrt{\pi}})-\frac{1}{2} y \phi(y) x\phi(x)\\&-\frac{1}{2} x \phi(x) y\phi(y)+\frac{1}{2} x y \phi(x) \phi(y)+\frac{\Phi(x)}{2}+\frac{\Phi(y)}{2}+\frac{1}{12}\\&-\frac{1}{2}(1-(1-\Phi(x))^2)-\frac{1}{2}(1-(1-\Phi(y))^2)\\
&=K_0(x,y)+\frac{1}{2}\Phi(x)(1-\Phi(x))+\frac{1}{2}\Phi(y)(1-\Phi(y))\\&+\frac{1}{2\sqrt{\pi}}(\varphi(x)+\varphi(y))-\varphi(x)\varphi(y)-\frac12xy\varphi(x)\varphi(y)+\frac{1}{12}-\frac{1}{4\pi}.
\end{align*}
\hfill$\Box$

The limiting distributions of EDF based test statistics are given in the following corollary.
\begin{corollary}\label{posledica}
Let $X_1,X_2,...,X_n$ be an i.i.d. sample from normal $\mathcal{N}(0,1)$. Then  we have that 

\begin{align*}
    \sqrt{n}D_n \overset{d}{\to}\sup_{t\in \mathbb{R}}|\eta(t)|;\\
    n\omega^2_n\overset{d}{\to}\sum_{i=1}^{\infty}\lambda_iZ_i^2;\\
    nA^2_n\overset{d}{\to}\sum_{i=1}^{\infty}\nu_iZ_i^2;\\
    \sqrt{n}G_n \overset{d}{\to}\sup_{t\in \mathbb{R}}|\xi(t)|;\\
    nU^2_n\overset{d}{\to}\sum_{i=1}^{\infty}\zeta_iZ_i^2
\end{align*}
where $Z_i$ are i.i.d. standard normal random variables, and $\{\lambda_i\}$, $\{\nu_i\}$ and $\{\zeta_i\}$ are sequences of eigenvalues of integral operators $\mathcal{W}$, $\mathcal{A}$ and $\mathcal{U}$ defined by
\begin{align}\label{operatorW}
    \mathcal{W}q(x)=\int_{-\infty}^{\infty}K_\eta(x,y)q(y)\varphi(y)dy,
\end{align}

\begin{align}\label{operatorA}
    \mathcal{A}q(x)=\int_{-\infty}^{\infty}\frac{K_\eta(x,y)}{\sqrt{\Phi(x)(1-\Phi(x))\Phi(y)(1-\Phi(y))}}q(y)\varphi(y)dy,
\end{align} and
\begin{align}\label{operatorU}
    \mathcal{U}q(x)=\int_{-\infty}^{\infty}K_\xi(x,y)q(y)\varphi(y)dy,
\end{align}
respectively.
\end{corollary}

For statistics $D_n$ and $G_n$ the convergence holds from the continuous mapping theorem, while for statistics $\omega^2_n$, $A_n$ and $U_n$ the proof follows from continuous mapping theorem,   Mercer's theorem  and  Karhunen-Loeve  decomposition of Gaussian process (see e.g. \cite{henze1997new}).

\section{Approximate Bahadur efficiency}

Let  $\mathcal{G}=\{G(x;\theta)\}$ be the family of distribution functions (DF's)  with densities $g(x;\theta)$, such that $G(x;\theta)$ is normal only for $\theta=0$.
We assume that the DF's from the class $\mathcal{G}$ satisfy the regularity conditions from \cite[Assumptions WD]{nikitinMetron}. 
%Denote $h(x)=g'_{\theta}(x;0)$ and $H(x)=G'_{\theta}(x;0)$.

Suppose that   $T_n=T_n(X_1,...,X_n)$ is a sequence of   test statistics where  the null hypothesis $H_0:\theta\in\Theta_0$ 
is rejected for $T_n>t_n$. Let the sequence of  DF's of the test statistic $T_n$ converge in distribution 
to a non-degenerate DF $F$. Additionally, suppose that

$$\log(1-F(t))=-\frac{a_Tt^2}{2}(1+o(1)),\;\;t\to \infty,$$
 and the limit in probability under the alternative
$$
\lim_{n\rightarrow\infty}T_n/\sqrt{n}=b_T(\theta)>0
$$ exists  for $\theta \in \Theta_1$.

The approximate relative  Bahadur efficiency with respect to another test statistic $V_n=V_n(X_1,...,X_n)$ is defined as
\begin{equation*}
e^{\ast}_{T,V}(\theta)=\frac{c^{\ast}_T (\theta)}{c^{\ast}_V (\theta)},
\end{equation*}
where
\begin{equation}\label{BASlope}
c^{\ast}_T(\theta)=a_Tb_T^2(\theta)
\end{equation} is the Bahadur  approximate slope of $T_n$. This is a  measure of the test efficiency proposed by Bahadur in \cite{bahadur1960}.

When studying asymptotic efficiency it is of interest to see the performance of tests for alternatives close to the null distribution.  For such alternatives we define the local approximate Bahadur efficiency by
\begin{align}
    e^{\ast}_{T,V}=\lim_{\theta\to0 }\frac{c^{\ast}_T (\theta)}{c^{\ast}_V (\theta)}.
\end{align}
The local approximate efficiency often coincides with the exact one.

Here we calculate the approximate relative Bahadur efficiency against some common close alternatives with respect to the likelihood ratio test (LRT).
The LRT has proven to be the optimal test in terms of the exact Bahadur efficiency, and is frequently used as a benchmark for comparison.

\subsection{Local Bahadur slope of the LRT for normality}

In \cite{bahadur1967} it was shown that the local exact Bahadur slope of LR test is equal to $2K(\theta)$ where $K(\theta)$ is the Kullback-Leibler distance from the alternative distribution indexed by $\theta$ to the family of null distributions. 
In the case of the null normality hypothesis it is equal to
 \begin{align}\label{KLopt}
     \nonumber K(\theta)&=\inf_{\mu, \sigma}{ \boldsymbol E}_{\theta}\log\frac{g(X,\theta)}{\frac{1}{\sigma}\varphi(\frac{X-\mu}{\sigma})}\\&=\inf_{\mu, \sigma}\int_{-\infty}^{\infty}
     \log\frac{g(x,\theta)}{\frac{1}{\sigma}\varphi(\frac{x-\mu}{\sigma})}g(x;\theta)dx,
 \end{align}
where $\varphi(x)$ is the standard normal density. 
In the case of close alternatives $g(x;\theta)$ its behaviour is given in the following theorem.

\begin{theorem}
For a given density $g(x;\theta)$ from $\mathcal{G}$ it holds
\begin{align}\label{Kulbak}
\begin{aligned}
    2K(\theta)&=\Bigg(\int_{-\infty}^{\infty}\frac{(g'_{\theta}(x;0))^2}{g(x;0)}dx-\frac{1}{\sigma_0^2}\Big(\int_{-\infty}^{\infty}xg'_{\theta}(x;0)dx\Big)^2\\&-\frac{1}{2\sigma_0^4}\Big(\int_{-\infty}^{\infty}(x-\mu_0)^2g'_{\theta}(x;0)dx\Big)^2\Bigg)\cdot\theta^2 + o(\theta^2),
    \end{aligned}
\end{align}
where $\mu_0$ and $\sigma^2_0$ are parameters of normal distribution $g(x;0)$.
\end{theorem}

\textbf{Proof}. The infimum in \eqref{KLopt} is reached for \begin{align}
    \mu(\theta)&=\int_{-\infty}^{\infty} x g(x;\theta)dx\\
    \sigma^2(\theta)&=\int_{-\infty}^{\infty} (x-\mu(\theta))^2 g(x;\theta)dx.
\end{align}

It is straightforward that $\mu(0)=\mu_0$, $\sigma^2(0)=\sigma^2_0$, as well as
\begin{align}\label{ex1}
    \mu'(0)&=\int_{-\infty}^{\infty} x g'_{\theta}(x;0)dx\\\label{ex2}
    \mu''(0)&=\int_{-\infty}^{\infty} x g''_{\theta}(x;0)dx\\\label{ex3}
    (\sigma^2)'(0)&=\int_{-\infty}^{\infty} (x-\mu_0)^2 g'_{\theta}(x;0)dx\\\label{ex4}
    (\sigma^2)''(0)&=-2\Big(\int_{-\infty}^{\infty} x g'_{\theta}(x;0)dx\Big)^2 + \int_{-\infty}^{\infty} (x-\mu_0)^2 g''_{\theta}(x;0)dx.
\end{align}

Differentiating $K(\theta)$ along $\theta$ with the help of expressions \eqref{ex1}-\eqref{ex4} we obtain that $K'(0)=0$ and $K''(0)$ equal to the right hand side of \eqref{Kulbak}. Expanding $K(\theta)$ in the Maclaurin series we complete the proof. \hfill$\Box$

The alternatives from $\mathcal{G}$ satisfy the conditions from \cite{rublik1989optimality} and hence the local approximate slope of LRT also has representations \eqref{Kulbak}.
\subsection{Local Bahadur slopes of the EDF based tests}

\begin{theorem}
For the statistics $D_n$, $\omega^2_n$, $A^2_n$, $G_n$ and $U_n^2$, and alternative density $g(x,\theta)\in \mathcal{G}$, the Bahadur approximate slopes are 

\begin{align*}
c_D(\theta)&=\frac{1}{\sup_x K_{\eta}(x,x)}\Big(\sup_{x}\big|g^\star(x)\big|\Big))^2\cdot\theta^2+o(\theta^2);\\
c_{\omega^2}(\theta)&=\frac{1}{\lambda_1}\int_{-\infty}^{\infty}\big(g^\star(x)\big)^2\varphi(x)dx\cdot\theta^2+o(\theta^2);\\
c_{A^2}(\theta)&=\frac{1}{\nu_1}\int_{-\infty}^{\infty}\frac{\big(g^\star(x)\big)^2}{\Phi(x)(1-\Phi(x))}\varphi(x)dx\cdot\theta^2+o(\theta^2);\\
c_{G}(\theta)&=\sup_{x\in\mathbb{R}}\Big|g^\star(x)-\int_{-\infty}^{\infty}\big(g^\star(u)\big)\varphi(u)du\Big|+o(\theta^2);\\
c_{U^2}(\theta)&=\int_{-\infty}^{\infty}\Big(g^\star(x)-\int_{-\infty}^{\infty}\big(g^\star(u)\big)\varphi(u)du\Big)^2\varphi(x)dx+o(\theta^2).
\end{align*}
respectively, where $\lambda_1$, $\nu_1$ and $\zeta_1$
are largest eigenvalues of operators $\mathcal{W}$,
$\mathcal{A}$ and $\mathcal{U}$ defined in {\eqref{operatorW}-\eqref{operatorU}}, and $$g^\star(x)=G'_{\theta}(x;0)+g(x;0)(\mu'(0)+x\sigma'(0)).$$
\end{theorem}

\begin{proof}
For each $x \in \mathbb{R}$, using the law of large numbers for U-statistics with estimated parameters \cite{iverson1989effects}, the limit in probability of $\eta_{n}(x,\hat{\mu},\hat{\sigma}^2)$ is
\begin{align*}
    B(x,\theta)=G(\mu(\theta)+\sigma(\theta)x,\theta)-\Phi(x)=\int_{-\infty}^{\mu(\theta)+\sigma(\theta)x}g(u,\theta)du-\Phi(x).
\end{align*}
Further we have that
\begin{align*}
    B_{\theta}'(x,\theta)=g(\mu(\theta)+\sigma(\theta)x,\theta)(\mu'(\theta)+\sigma'(\theta)x)+\int_{-\infty}^{\mu(\theta)+\sigma(\theta)x}g_{\theta}'(u,\theta)du
\end{align*}
When $\theta=0$ the expression above is equal to
\begin{align*}
  B_{\theta}'(x,0)=g(x,0)(\mu'(0)+\sigma'(0)x)+G'_\theta(x;0).  
\end{align*}
Hence we obtain that
\begin{align*}
    B(x,\theta)=g^\star(x)\cdot\theta+o(\theta),\;\;\theta\to0.
\end{align*}
Following \cite[Chap. 19]{van2000asymptotic}, the limits in $P_{\theta}$ probability of statistics $D_n$, $W_n$ and $A_n$ are then
\begin{align*}
    b_D(\theta)&=\sup_{x}|g^\star(x)|\cdot\theta+o(\theta);\\
    b_{\omega^2}(\theta)&=\int_{-\infty}^{\infty}(g^\star(x))^2\varphi(x)dx\cdot\theta^2+o(\theta^2);\\
    b_{A^2}(\theta)&=\int_{-\infty}^{\infty}\frac{(g^\star(x))^2}{\Phi(x)(1-\Phi(x))}\varphi(x)dx\cdot\theta^2+o(\theta^2).
\end{align*}

Analogously, using the process $\xi_n(x;\hat{\mu},\hat{\sigma}^2)$, we obtain the 
limits in probability of $G_n$ and $U^2_n$ are
\begin{align*}
    b_{G}(\theta)&=\sup_{x}\Big|g^\star(x)-\int_{-\infty}^{\infty}\big(g^\star(u)\big)\varphi(u)du\Big|\cdot\theta+o(\theta);\\
b_{U^2}(\theta)&=\int_{-\infty}^{\infty}\Big(g^\star(x)-\int_{-\infty}^{\infty}\big(g^\star(u)\big)\varphi(u)du\Big)^2\varphi(x)dx\cdot\theta^2+o(\theta^2).
\end{align*}

The tail behaviour of the supremum of a Gaussian process follows from  \cite{marcusShepp}, and the constant $a_T$ from \eqref{BASlope} is equal to the supremum on the diagonal of the covariance function. Therefore we get $a_D=\sup_{t}K_{\eta}(t,t)$ in the case of $D_n$ and $a_G=\sup_{t}K_{\xi}(t,t)$ in the case of $G_n$.

For the integral type statistic $\omega^2_n$, using the result of Zolotarev \cite{zolotarev}, we have that the logarithmic tail behavior of $\widetilde{\omega}^2=\sqrt{n\omega^2_n}$ is
\begin{align*}
\log(1-F_{\widetilde{\omega}^2}(x))=-\frac{x^2}{2\lambda_1}+o(x^2),\;\; x\to \infty,
\end{align*}
and hence, $\widetilde{a}_{\widetilde{\omega}^2}=\frac{1}{\lambda_1}$, where $\lambda_!$ is the largest eigenvalue of the integral operator $\mathcal{W}$ defined in \eqref{operatorW}. Analogously we get $\widetilde{a}_{\widetilde{A}^2}=\frac{1}{\nu_1}$
and $\widetilde{a}_{\widetilde{U}^2}=\frac{1}{\zeta_1}$ for statistics $A^2_n$ and $U^2_n$.
\end{proof}

\subsection{Calculation of efficiencies}

The close alternatives we consider here are

\begin{itemize}
    \item a Lehmann alternative with density
    \begin{equation*}
        g_1(x;\theta)=(1+\theta)\Phi^{\theta}(x)\varphi(x);
    \end{equation*}
    \item a first Ley-Paindaveine alternative with density
    \begin{equation*}
        g_2(x;\theta)=\varphi(x) e^{-\theta(1-\Phi(x))}(1+\theta\Phi(x));
    \end{equation*}
    \item a second Ley-Paindaveine alternative with density
    \begin{equation*}
        g_3(x;\theta)=\varphi(x)(1-\theta\pi\cos(\pi\Phi(x));
    \end{equation*}
    \item a contamination alternative (with $\mathcal{N}(\mu,\sigma^2)$) alternative with density
    \begin{equation*}
        g_4^{[m,\sigma^2]}(x;\theta)=(1-\theta)\varphi(x)+\frac{\theta}{\sigma}\varphi\Big(\frac{x-\mu}{\sigma}\Big).
    \end{equation*}
    %\item Skew normal (Azzalini)
    %\begin{align}
        %g_5(x;\theta)=2\Phi(\theta x)\varphi(x)
    %\end{align}
\end{itemize}

To calculate the efficiency one needs to find the largest eigenvalues $\lambda_1$, $\nu_1$ and $\zeta_1$ from Corollary \ref{posledica}. Since we can not obtain them analytically, we use the approximation method from \cite{bozin} (see also \cite{cuparic2020some}).

The values of efficiencies are presented in Table \ref{tab: eff}. We can see that the integral tests are more efficient than the supremum ones. Among them, the Anderson--Darling test is best one for almost all considered alternatives. Additionally, the Watson-type modifications of Kolmogorov-Smirnov and Cramer--von Mises tests are less efficient than the original versions. 

These results can serve as a benchmark for evaluation of the quality of recent and future normality tests.

\begin{table}
\centering
\caption{Approximate Bahadur efficiency of $D_n$ and $W_n$ with respect to LRT}
\label{tab: eff}
\medskip

\begin{tabular}{c|ccccc}
     alternative &$D_n$&$\omega^2_n$&$A^2_n$&$G_n$&$U^2_n$  \\\hline
   Lehmann&0.311&0.584& 0.689&0.258&0.471\\
   1st Ley-Paindaveine&0.455&0.800& 0.891&0.321&0.699\\
   2nd Ley-Paindaveine&0.565&0.917& 0.971&0.332&0.846\\ 
   Contamination with $\mathcal{N}(1,1)$ &0.200&0.377&0.464&0.111&0.302\\
   Contamination with $\mathcal{N}(0.5,1)$&0.266&0.505&0.606&0.146&0.402\\
   Contamination with $\mathcal{N}(0,0.5)$&0.258&0.570&0.649&0.137&0.668\\
\end{tabular}
\end{table}

\section*{Acknowledgement}
The work of B. Milo\v sevi\'c is supported by the Ministry of education, science and technological development of the Republic of Serbia.
%\begin{thebibliography}{3}
%\bibitem{bahadur1960} Bahadur, R. R. (1960). Stochastic comparison of tests. %Annals of Mathematical Statistics, 31(2), 276--295.
%\bibitem{bahadur1967} Bahadur, R. R. (1967). Rates of convergence of estimates %and test statistics. The Annals of Mathematical Statistics, 38(2), 303-324.
%\bibitem{billingsley} Billingsley, P. (1999). Convergence of probability measures. John Wiley \& Sons.
%\bibitem{bozin} Bo\v zin, V., Milo\v sevi\'c, B., Nikitin, Ya.Yu.,  Obradovi\'c, M. (2020). New Characterization-Based Symmetry Tests. Bulletin of the Malaysian Mathematical Sciences Society, 43(1),297--320.
%\bibitem{kato}Kato, T. (1980). Perturbation theory for linear operators (2nd Edition). Springer Verlag, Berlin Heidelberg New York.
%\bibitem{marcusShepp} Marcus, M. B.,  L.A. Shepp (1972). Sample behavior of Gaussian processes. In Proc. of the Sixth Berkeley Symposium on Mathematical Statistics and Probability (Vol. 2,  423-441)
%\bibitem{nikitinMetron} Nikitin, Ya.Yu. and Peaucelle, I.  (2004). Efficiency and local optimality of nonparametric tests based on U-and V-statistics. Metron, 62(2), 185-200.
%\bibitem{zolotarev} Zolotarev, V. M. (1961). Concerning a certain probability problem. Theory of Probability \& Its Applications, 6(2), 201-204.

%\end{thebibliography}
%\bibliographystyle{amsplain}
%\bibliography{literatura}

\begin{thebibliography}{10}

\bibitem{anderson1952}
T.~W. Anderson and D.~A. Darling, \emph{Asymptotic theory of certain "goodness
  of fit" criteria based on stochastic processes}.--- Ann. Math. Stat. \textbf{23}, No.~2 (1952),  193--212.

\bibitem{arcones2006bahadur}
M.~A. Arcones, \emph{On the {B}ahadur slope of the {L}illiefors and the
  {C}ram{\'e}r-von {M}ises tests of normality}. --- Lecture Notes-Monograph Series
  \textbf{51} (2006), 196--206.

\bibitem{bahadur1960}
R.~R. Bahadur, \emph{Stochastic comparison of tests}. --- Ann. Math. Stat. \textbf{31}, No.~2 (1960), 276--295.

\bibitem{bahadur1967}
R.~R. Bahadur, \emph{Rates of convergence of estimates and test statistics}. --- Ann. Math. Stat. \textbf{38}, No.~2 (1967), 303--324.

\bibitem{billingsley}
P.~Billingsley, \emph{Convergence of probability measures}, John Wiley \&
  Sons., 1999.

\bibitem{bozin}
V.~Bo\v{z}in, B.~Milo\v{s}evi\'c, \relax{Ya.~Yu}. Nikitin, and
  M.~Obradovi{\'c}, \emph{New characterization-based symmetry tests}. --- Bull. Malaysian Math. Sci. Soc. \textbf{43}, No.~1 (2020),
  297--320.

\bibitem{cuparic2020some}
M.~Cupari{\'c}, B.~Milo{\v{s}}evi{\'c}, \relax{Ya.~Yu}. Nikitin, and
  M.~Obradovi{\'c}, \emph{Some consistent exponentiality tests based on
  {P}uri--{R}ubin and {D}esu characterizations}. --- Appl. Math. Prague
  \textbf{65}, No.~3 (2020), 245--255.

\bibitem{darling1955cramer}
D.~A. Darling, \emph{The {C}ramer-{S}mirnov test in the parametric case}. --- Ann. Math. Stat.  \textbf{26}, No.~1 (1955), 1--20.

\bibitem{darling1957}

D.~A. Darling, \emph{The {K}olmogorov-{S}mirnov, {C}ramer-von {M}ises tests}. --- Ann. Math. Stat.  \textbf{28}, No.~4 (1957), 823--838.

\bibitem{darling1983a}
D.~A. Darling, \emph{On the asymptotic distribution of {W}atson's statistic}. --- Ann. Stat \textbf{11} (1983), No.~4, 1263--1266.

\bibitem{darling1983b}
D.~A. Darling, \emph{On the supremum of a certain {G}aussian process}. --- Ann.
  Probab. \textbf{11} (1983), no.~3, 803--806.

\bibitem{durbin1973weak}
J.~Durbin, \emph{Weak convergence of the sample distribution function when
  parameters are estimated}. --- Ann. Stat. \textbf{1}, No.~2 (1973),
  279--290.

\bibitem{henze1997new}
N.~Henze and T.~Wagner, \emph{A new approach to the {BHEP} tests for
  multivariate normality}. --- J. Multivar. Annal. \textbf{62}, No.~1 (1997),
   1--23.


\bibitem{iverson1989effects}
H.~K. Iverson and R.~H. Randles, \emph{The effects on convergence of substituting parameter estimates
  into {U}-statistics and other families of statistics}. --- Probab. Theory Relat. Fields \textbf{81}, No.~3 (1989), 453--471.

\bibitem{kac1955tests}
M.~Kac, J.~Kiefer, and J.~Wolfowitz, \emph{On tests of normality and other
  tests of goodness of fit based on distance methods}. --- Ann. Math. Stat. \textbf{26}, No.~2 (1955), 189--211.

\bibitem{khmaladze1982}
E.~V. Khmaladze, \emph{Martingale approach in the theory of goodness-of-fit
  tests}. ---  Theory Probab. its Appl. \textbf{26}, No.~2 (1982),
  240--257.

\bibitem{kolmogorov1933}
A.~N. Kolmogorov, \emph{Sulla determinazione empirica di una legge di
  distribuzione}. --- Giorn. Ist. Ital. Attuari \textbf{4}, No.~1
  (1933), 83--91.

\bibitem{kuiper1960}
N.~H. Kuiper, \emph{Tests concerning random points on a circle}. --- Proc. Sect. Sci. K. Ned. Akad. Wet. Amst. Series A. \textbf{63}, No.~1 (1960), 38--47.

\bibitem{lilliefors1967kolmogorov}
H.~W. Lilliefors, \emph{On the {K}olmogorov-{S}mirnov test for normality with
  mean and variance unknown}. --- J. Am. Stat. Assoc.
  \textbf{62}, No.~318 (1967), 399--402.

\bibitem{lilliefors1969kolmogorov}
H.~W. Lilliefors, \emph{On the {K}olmogorov-{S}mirnov test for the exponential
  distribution with mean unknown}. --- J. Am. Stat. Assoc. \textbf{64}, No.~325 (1969), 387--389.

\bibitem{marcusShepp}
M.~B. Marcus and L.~A. Shepp, \emph{Sample behavior of {G}aussian processes}. ---
  Proc. of the Sixth Berkeley Symposium on Mathematical Statistics and
  Probability, \textbf{2}, 1972, 423--441.

\bibitem{milovsevic2016two}
B.~Milo{\v{s}}evi{\'c} and M.~Obradovi{\'c}, \emph{Two-dimensional
  {K}olmogorov-type goodness-of-fit tests based on characterisations and their
  asymptotic efficiencies}. --- J. Nonparametr. Stat.
  \textbf{\textbf{28}}, No.~2 (2016), 413--427.

\bibitem{nikitinKnjiga}
\relax{Ya.~Yu}. Nikitin, \emph{Asymptotic efficiency of nonparametric tests},
  Cambridge University Press, 1995.

\bibitem{nikitin2010large}
\relax{Ya.~Yu}. Nikitin,\emph{Large deviations of {U}-empirical {K}olmogorov--{S}mirnov tests
  and their efficiency}. --- J. Nonparametr. Stat. \textbf{22}, No.~5
  (2010), 649--668.

\bibitem{nikitinMetron}
\relax{Ya.~Yu}. Nikitin and I.~Peaucelle, \emph{Efficiency and local optimality
  of nonparametric tests based on {U}-and {V}-statistics}. --- Metron \textbf{LXII}, No.~2
  (2004), 185--200.

\bibitem{nikitin1999large}
\relax{Ya.~Yu}. Nikitin and E.~V. Ponikarov, \emph{Large deviations of
  {C}hernoff type for {U}-and {V}-statistics}. --- Dokl. Math.  \textbf{60}, No.~3
  (1999), 316--318.

\bibitem{nikitin2007lilliefors}
\relax{Ya.~Yu}. Nikitin and A.~V. Tchirina, \emph{Lilliefors test for
  exponentiality: large deviations, asymptotic efficiency, and conditions of
  local optimality}.--- Math. Methods Stat. \textbf{16},
  No.~1 (2007), 16--24.

\bibitem{randles1982asymptotic}
R.~H. Randles, \emph{On the asymptotic normality of statistics with estimated
  parameters}. --- Ann. Stat. \textbf{10}, No.~2 (1982), 462--474.

\bibitem{rublik1989optimality}
F.~Rubl{\'\i}k, \emph{On optimality of the {LR} tests in the sense of exact
  slopes. {I.} {G}eneral case}. --- Kybernetika \textbf{25}, No.~1 (1989), 13--14.

\bibitem{stephens1976asymptotic}
M.~A. Stephens, \emph{Asymptotic results for goodness-of-fit statistics with
  unknown parameters}. --- Ann. Stat. \textbf{4}, No.~2 (1976),
  357--369.

\bibitem{sukhatme1972}
S.~Sukhatme, \emph{Fredholm determinant of a positive definite kernel of a
  special type and its application}. --- Ann. Math. Stat. 
  \textbf{43}, No.~6 (1972), 1914--1926.

\bibitem{van2000asymptotic}
A.W. Van~der Vaart, \emph{Asymptotic statistics}, vol.~3, Cambridge University
  Press, 2000.

\bibitem{watson1961}
G.~S. Watson, \emph{Goodness-of-fit tests on a circle}. --- Biometrika \textbf{48}, No.~1/2
  (1961),  109--114.

\bibitem{watson1976}
G.~S. Watson,  \emph{Optimal invariant tests for uniformity}. --- Studies in Probability
  and Statistics, Papers in honour of {E}.{J}.{G}. {P}itman (E.~J. Williams,
  ed.), Amsterdam: North-Holland, 1976, 121--127.

\bibitem{zolotarev}
V.~M. Zolotarev, \emph{Concerning a certain probability
  problem}. --- Theory Probab. its Appl. \textbf{6}, No.~2 (1961),
  201--204.


\end{thebibliography}
\providecommand{\bysame}{\leavevmode\hbox to3em{\hrulefill}\thinspace}
\providecommand{\MR}{\relax\ifhmode\unskip\space\fi MR }
\providecommand{\MRhref}[2]{%
  \href{http://www.ams.org/mathscinet-getitem?mr=#1}{#2}
}
\providecommand{\href}[2]{#2}

\end{document}